\documentclass[twocolumn]{autart}  

\usepackage[utf8]{inputenc}
\usepackage[english]{babel}
\usepackage{graphicx,overpic,tikz}     
\usepackage{amsmath}
\usepackage{amssymb}
\usepackage{color}
\usepackage{float}
\usepackage{mathrsfs}
\usepackage{hyperref}
\usepackage{subfigure} %%I have added this package

\usepackage{pgfplots}
\makeatletter
\let\theoremstyle\@undefined                        % undefine \theoremstyle
\makeatother

\usepackage{amsthm}

\usepackage{mathptmx}
\def\epsilon{\varepsilon}

\def\R{\mathbb{R}}

\newtheorem{Theo}{Theorem}
\newtheorem{Pro}{Proposition}
\newtheorem{Defi}{Definition}

\newtheorem{Lem}{Lemma}
\newtheorem{remark}{Remark}
\newtheorem{corollary}{Corollary}

\usepackage{floatflt}

\begin{document}
%\sloppy
\begin{frontmatter}

\title{ Event-triggered boundary control of constant-parameter reaction–diffusion
PDEs: a small-gain approach} % Title, preferably not more 
                                                % than 10 words.

\thanks[footnoteinfo]{This work has been partially supported  by CPER ELSAT 2020, Project,
Contratech.
%Corresponding author: {\tt\footnotesize nicolas.espitia-hoyos@inria.fr }
}

\author[INRIA,Central]{Nicol\'as Espitia}\ead{nicolas.espitia-hoyos@inria.fr}, \ %
\author[NATUA]{Iasson Karafyllis}\ead{iasonkar@central.ntua.gr,iasonkaraf@gmail.com}, \ 
\author[USD]{Miroslav Krstic}\ead{krstic@ucsd.edu}
%\author[gipsa,UGA]{Christophe Prieur}\ead{ christophe.prieur@gipsa-lab.fr}               % e-mail address  % (ead) as shown

\address[INRIA]{Valse team, Inria
Lille - Nord Europe }  % Please supply                                              
\address[Central]{CRIStAL (UMR-CNRS 9189), Ecole Centrale de Lille, Cit\'e Scientifique, 59651 Villeneuve-dAscq, France.}
\address[NATUA]{Department of Mathematics, National Technical University of Athens, Zografou Campus, 15780, Athens, Greece.}             % full addresses
\address[USD]{Department of Mechanical and Aerospace Engineering, University of California, San Diego, La Jolla, CA 92093-0411, USA.}  
%\title{Event-based control of  linear hyperbolic systems \\
%of conservation laws }
%\author[gipsa]{Nicol\'as Espitia}\quad
%\author[ljk]{Antoine Girard}\quad
%\author[gipsa]{Nicolas Marchand}\quad
%\author[gipsa]{Christophe Prieur}\quad
%
%\address[gipsa]{Gipsa-lab, Grenoble Campus, 11 rue des Math\'ematiques, BP 46,
%38402 Saint Martin d'H\`eres Cedex, France}
%\address[ljk]{Laboratoire Jean Kuntzmann, Universit\'e de Grenoble, BP 53, 38041 Grenoble, France}
%\thanks{E-mail adresses: ying.tang@gipsa-lab.fr,\\ christophe.prieur@gipsa-lab.fr, antoine.girard@imag.fr}


\begin{abstract}                          % Abstract of not more than 200 words.
This paper deals with an  event-triggered boundary control of constant-parameters reaction-diffusion PDE systems. The approach relies on   the emulation of  backstepping control along with   a suitable triggering condition which establishes the time instants at which the control value needs to be sampled/updated. 
In this paper,  it is shown that under the proposed event-triggered boundary control, there exists a  minimal dwell-time (independent of the initial condition) between two triggering times and furthermore the well-posedness and global exponential stability are  guaranteed. The analysis follows small-gain arguments and builds on  recent papers  on sampled-data control for this kind of PDE.   A simulation example is presented to validate the theoretical results.
\end{abstract}

\end{frontmatter}
\section{Introduction}
Control and estimation strategies must be implemented and validated into digital platforms. It is
important to study carefully the  issues concerning digital control such as sampling. This is because, if sampling is not addressed properly, the stability and estimation properties may be  lost. For finite-dimensional systems, namely networked control systems modeled by ordinary differential equations (ODEs), digital control has been extensively developed and several schemes for discretization and for sampling in time continuous-time controllers have been investigated, e.g., by sampled-data control  \cite{Hetel2017309} and event-triggered control strategies
 \cite{tabuada2007event,event-triger-Heeme-Johan-Tabu,Marchand2013-universal-formula,lemmon2010event,girard2014dynamic,Postoyan_Aframework_ETS2014,JiangSmallGainETC,Liu_ZPJiang2015}. The latter has become  popular and promising due to not only its efficient way of using communication and
computational resources by updating the control value aperiodically (only when needed) but also due to its rigorous
way of implementing continuous-time controllers into digital platforms. 

In general,  event-triggered control includes two main components: a feedback control law which  stabilizes the system  and an event-triggered mechanism which contains a triggering condition  that determines the time instants at which the control needs to be updated.  Two general approaches  exist for the design:  \textit{Emulation}  from which  the controller is a priori predesigned and only the event-triggered algorithm has to be designed (as in e.g. \cite{tabuada2007event})  and    \textit{Co-design}, where the joint design of the control law and the event-triggering mechanism  is   performed simultaneously (see e.g.  \cite{Seuret2016}).

 Nevertheless, for partial differential equations (PDEs) sampled-data  and event-triggered control strategies  without  model reduction   have not  achieved a sufficient level of maturity as in the finite-dimensional case. It has not been sufficiently clear (from theoretical and practical point of view)  how fast  sampling the in-domain or the boundary continuous-time controllers  should be for  preserving both stability and convergence properties of PDE systems. 
Few approaches on sampled-data and event-triggered control of parabolic PDEs  are  considered in \cite{Fridman2012826,KARAFYLLIS2018226,Selivanov_FridmanAuto}, \cite{Yao2013,JIANG20162854}.  In the context of abstract formulation of distributed parameter systems, sampled-data control is investigated  in \cite{Logemann2005} and \cite{Tan2009}.
 For hyperbolic PDEs,  sampled-data control is  studied in \cite{Davosampling2018} and  \cite{Sampled_dataKarafylis_Kristic}. 
Some recent works  have introduced  event-triggered  control strategies for linear hyperbolic PDEs under an emulation approach \cite{Espitia2016_Aut,Espitia2016Nolcos,Espitia2018TAC}. In  \cite{Espitia2016_Aut}  and \cite{Espitia2016Nolcos}, for instance,   event-triggered boundary controllers for linear conservation laws using output feedback are studied by following Lyapunov techniques (inspired by \cite{BastinCoron_book2016}).    In \cite{Espitia2018TAC},  the approach relies on  the backstepping method for coupled system of balance laws (inspired by \cite{Vazquez2011,krstic2008backstepping}) which leads to a full-state feedback control which is sampled according to a dynamic triggering condition. % (inspired by the ones originally introduced for finite-dimensional systems in \cite{girard2014dynamic} and for hyperbolic PDEs in \cite{Espitia2016Nolcos}). 
 Under such a triggering policy,  it has been possible to prove the existence of a minimal dwell-time between triggering time instants and therefore avoiding the so-called Zeno phenomena.\\
  
In sampled-data  control as well as   in  event-triggered control scenarios for PDEs, the effect of sampling (and therefore, the  underlying actuation error) has to be carefully handled. In particular,  for    reaction-diffusion parabolic PDEs    the  situation of having such errors  at the boundaries has been challenging and has become a central issue; especially  when having Dirichlet boundary conditions due to the lack of an ISS-Lyapunov function for the stability analysis.  In \cite{KARAFYLLIS2018226} this problem has been overcome by studying  ISS properties   directly from the nature of the PDE system (see also  e.g. \cite{Karafyllis_Krstic_SIAM2019,Karafyllis2019_book}) while using  modal decomposition and Fourier series analysis.  Lyapunov-based  approach  has not  been necessary   to perform the stability analysis and to be able to come up with  ISS properties and  small gain arguments.  Thus,  it has been possible to establish  the robustness with respect to  the actuation error.  This approach has allowed the derivation  of  an estimate of the diameter  of the  sampling period  on which the control is updated in a sampled-and-hold fashion. The drawback, however, is that such a period turns out to be truly  small, rendering   the approach very conservative. With periodic implementation, one may produce unnecessary updates of the sampled controllers, which   cause over utilization of computational and communication resources, as well as actuator  changes that are more frequent than necessary.

This issue  strongly motivates the study  of  event-triggered control for PDE systems. %.    Hence, we believe that   event-triggered control may  show benefits with respect to periodic schemes as the actuation updating is done only when needed. In overall, event-triggered would represent a more realistic approach for the actuation on the PDE system.
Therefore,  inspired by \cite{KARAFYLLIS2018226},   in this paper we propose an event-triggered boundary control based on the emulation of the backstepping boundary control. An event-triggering condition is derived and the stability analysis is performed by using small-gain arguments.  

The main contributions  are  summed up  as follows: 
\begin{itemize}
\item We prove that under the event-triggered control no Zeno solutions can appear.  A uniform minimal dwell-time (independent of the initial condition)   between two consecutive triggering time instants has been obtained.
\item Consequently, we guarantee the existence and uniqueness of solutions to the closed-loop system.
\item We prove that under the event-triggered boundary control,  the closed-loop system  is globally exponentially stable in the $L^2$- norm sense.
\end{itemize} 

%There are several   motivations of using event-triggered control in the PDE setting. {\color{red}One of them is precisely because hyperbolic PDEs have been useful for the modeling and control of variety of  physical networks: e.g.  hydraulic  \cite[Chapter 8]{BastinCoron_book2016}, communication \cite{Espitia_communicationnetworks}, and road    traffic networks \cite{coclite_garavelllo_picolli_2005}. Therefore, either the boundary or the in-domain control must   be implemented in digital platforms while suitably sampling the control value. }  Another reason is related to the efficient use of resources and actuation solicitation. For instance,  actuation on   hydraulic  networks of channel flows   may be expensive due to  actuators inertia when regulating the water level and the water flow rate by using gates opening as control actions. Then, with  event-triggered control one may  modulate efficiently   the gates opening, only when is truly necessary. The difference with respect to   periodic sampling schemes is evident. With periodic implementation, one may produce unnecessary updates of the sampled controllers, which will  cause high utilization of computational and communication resources, as well as actuator solicitation. Hence, we believe that   event-triggered control may  show benefits with respect to periodic schemes. In overall, event-triggered would represent a more realistic approach for the actuation on the network.

The paper is organized as follows. In Section \ref{section_problem_form}, we introduce the class of reaction-diffusion parabolic systems,  some preliminaries on stability and backstepping   boundary control and the  preliminary notion of existence and uniqueness of solutions. Section \ref{Event-triggered-strategySection}  provides the  event-triggered boundary control and the main results.   Section~\ref{numerical_simulation} provides a numerical example to illustrate the main results.  Finally, conclusions and perspectives are given in Section~\ref{conslusion_and_perspect}.

\paragraph*{Notations}
$ \R_{+}$ will denote the set of nonnegative real numbers.   Let $S \subseteq \mathbb{R}^n$ be an open set and let $A \subseteq  \mathbb{R}^n$ be a set that satisfies $S \subseteq  A \subseteq  \bar{S}$. By $C^0(A;\Omega)$, we denote the class of continuous functions on $A$, which take values in $\Omega \subseteq  \mathbb{R}$. By $C^k(A;\Omega)$, where $k\geq 1$ is an integer, we denote the class of functions on $A$, which takes values in $\Omega$ and has continuous derivatives of order $k$. In other words, the functions of class $C^k(A;\Omega)$ are the functions which have continuous derivatives of order $k$ in $S=int(A)$ that can be continued continuously to all points in $\partial S \cap A$.  $L^2(0,1)$ denotes the equivalence class of Lebesgue measurable functions $f: [0,1] \rightarrow \R$ such that $\Vert f\Vert =\left(\int_{0}^{1}  \vert f(x)\vert ^{2}dx\right)^{1/2} < \infty$.   Let $u:\mathbb{R}_{+}\times [0,1] \rightarrow \mathbb{R}$ be given. $u[t]$ denotes the profile of $u$ at certain $t\geq 0$, i.e. $(u[t])(x) = u(t,x)$, for all $x \in [0,1]$. For an interval $I \subseteq \mathbb{R}_{+}$, the space $C^{0}(I;L^2(0,1))$ is the space of continuous mappings $I\ni t \rightarrow u[t] \in L^2(0,1)$. $H^2(0,1)$  denotes the Sobolev space of functions $f \in L^2(0,1)$ with square integrable (weak) first and second-order derivatives $f^{'}(\cdot), f^{''}(\cdot) \in L^2(0,1)$.  $I_m(\cdot)$, $J_m(\cdot)$ with $m \in \mathbb{Z}$, denote the modified Bessel and (nonmodified) Bessel functions of the first kind.

\section{Preliminaries and problem description}\label{section_problem_form}

Let us consider the following  scalar reaction-diffusion system with constant coefficients:
\begin{eqnarray}\label{eq:sysparabolic0}
u_t(t,x) & =&  \theta u_{xx}(t,x) + \lambda u(t,x) \\
u(t,0)&=&0\\
u(t,1)&=& U(t)  \label{BC_parabolic_PDE_u0}
\end{eqnarray} 
and initial condition:
\begin{equation}\label{IC_parabolic_PDE_u0}
 u(0,x)=u_{0}(x)
\end{equation}
where $\theta >0$ and $\lambda \in \mathbb{R}$. $u: [0,\infty)\times[0,1]  \rightarrow \mathbb{R} $ is the system state  and  
   $U(t) \in \mathbb{R}$ is the control input. The control design relies on the Backstepping approach \cite{Smyshlyaev-Krstic2004,krstic2008boundary}  under which   the following continuous-time   controller  (nominal boundary feedback)  has been obtained:
\begin{equation}\label{control_function_continuous_with_K}
U(t)=  \int_{0}^{1} K(1,y)u(t,y)dy  
\end{equation} 
It has then been proved that  the  under  continuous-time controller \eqref{control_function_continuous_with_K}  with 
control gain $K$ satisfying:
\begin{equation}\label{solutionKernelexplicit}
\begin{split}
 K(x,y)& = -y\gamma\frac{I_1\left(\sqrt{\gamma(x^2- y^2)}\right)}{ \sqrt{\gamma(x^2- y^2)}}
\end{split}
\end{equation}
evolving in  a triangular domain given by $\mathcal{T}= \{  (x,y): 0 \leq y < x \leq 1 \}$ and  with $\gamma=(\lambda +c)/\theta$ (where $c\geq 0$ is a design parameter),   the closed-loop system \eqref{eq:sysparabolic0}-\eqref{IC_parabolic_PDE_u0} is globally exponentially stable  in $L^2$- norm sense. 

\subsection{ \textbf{Event-triggered control and emulation of the backstepping design}}

We aim at stabilizing the closed-loop system on events while   sampling  the continuous-time controller \eqref{control_function_continuous_with_K}   at certain sequence of time instants  $(t_{j})_{j \in \mathbb{N}}$,   that will be characterized later on.   The control value is held constant between two successive time instants and it is updated when some state-dependent condition is verified.
In this scenario,   we need to suitably modify  the boundary condition in \eqref{eq:sysparabolic0}-\eqref{BC_parabolic_PDE_u0}.
The boundary value of the state   is going to be  given by: 
\begin{equation}\label{new_boundary_originalsystems}
u(t,1) = U_d(t) 
\end{equation}
with 
\begin{equation}\label{control_function_event-triggered_with_K}
U_d(t)=  \int_{0}^{1} K(1,y)u(t_{j},y)dy  
\end{equation}
for all $t \in [t_{j},t_{j+1})$, $j \geq 0$. Note that $U_d(t)=U(t)+d(t)$ with $U(t)$ given by \eqref{control_function_continuous_with_K} and $d$ given by:
\begin{equation}\label{deviation_actuation}
d(t)=\int_{0}^{1} K(1,y)u(t_{j},y) dy   - \int_{0}^{1} K(1,y) u(t,y)  dy 
\end{equation}
Here,  $d$ (which will be fully characterized along with $(t_{j})_{j\in \mathbb{N}}$ in the next section) can be viewed as an actuation deviation between the   nominal boundary feedback  and the event-triggered boundary control \footnote{In sampled-data control  as in  \cite{KARAFYLLIS2018226}, such a deviation is called  input holding error.}.

Hence, the control problem we aim at handling is the following:
\begin{eqnarray}\label{eq:sysparabolic}
u_t(t,x) & =&  \theta u_{xx}(t,x) + \lambda u(t,x) \\
u(t,0)&=&0\\
u(t,1)&=& U_d(t)  \label{BC_parabolic_PDE_u}
\end{eqnarray} 
for all $t \in [t_{j},t_{j+1})$, $j\geq 0$,  and initial condition:
\begin{equation}\label{IC_parabolic_PDE_u}
 u(0,x)=u_{0}(x)
\end{equation}
We will perform the emulation of the backstepping  which requires also information of the target system. Indeed, let us recall that the backstepping method  makes use of an invertible Volterra transformation:
\begin{equation}\label{backstepping _inverse_trasf_1}
\begin{split}
w(t,x) &=  u(t,x) - \int_{0}^{x} K(x,y)u(t,y) dy 
\end{split}
\end{equation}
with kernel $K(x,y)$ satisfying \eqref{solutionKernelexplicit} which    maps the  system \eqref{eq:sysparabolic}-\eqref{IC_parabolic_PDE_u}   into the  following  target system:
\begin{eqnarray}\label{Target_system_w}
w_t(t,x) & =&  \theta w_{xx}(t,x) - c w(t,x)  \\
w(t,0)&=&0\\
w(t,1)&=& d(t)
\end{eqnarray} 
 with initial condition:
\begin{equation}\label{IC_Target_system_w}
 w(0,x)=u_{0}(x)  -  \int_{0}^{x} K(x,y)u_{0}(y) dy 
\end{equation}
where $c>0$ can be chosen arbitrary.\\

\begin{remark}
It is worth recalling that the  Volterra backstepping transformation \eqref{backstepping _inverse_trasf_1} is invertible whose inverse is given as follows: \\
\begin{equation}\label{backstepping _trasf_inverse}
\begin{split}
 u(t,x) &=  w(t,x) + \int_{0}^{x} L(x,y)w(t,y) dy 
\end{split}
\end{equation} 
where $L $ satisfies:
\begin{equation}\label{solutionKernelexplicit_inverse}
\begin{split}
L(x,y)& = -y\gamma\frac{J_1\left(\sqrt{\gamma(x^2- y^2)}\right)}{ \sqrt{\gamma(x^2- y^2)}}
\end{split}
\end{equation}
 with $\gamma=(\lambda +c)/\theta$.% where $c>0$ is  chosen sufficiently large.
%and fulfills the following property
\end{remark}

\subsection{\textbf{Well-posedness issues}}
  The notion of solution for 1-D  linear parabolic systems under boundary sampled-data control has been rigorously analyzed in  \cite{KARAFYLLIS2018226}. In this paper, we follow the same framework. \\

\begin{Pro}\label{existence_and_continuity_solutions}
 There exists a unique solution  $u \in  \mathcal{C}^{0}([t_{j},t_{j+1}]; L^{2}(0,1))$ to the system \eqref{eq:sysparabolic}-\eqref{IC_parabolic_PDE_u}   between two time instants $t_{j}$ and $t_{j+1}$ satisfying $u \in C^{1}((t_{j},t_{j+1}) \times [0,1])$, $u[t] \in C^2([0,1])$ for all $t \in (t_{j},t_{j+1}] $ and initial data $u[t_{j}] \in L^2(0,1)$.  
\end{Pro}
\begin{proof}
It is a straightforward application of  \cite[Theorem 2.1]{KARAFYLLIS2018226}. 
\end{proof}

 In what follows we assume that in  open-loop, the  system  \eqref{eq:sysparabolic}-\eqref{IC_parabolic_PDE_u} is unstable or neutrally stable, i.e., $\lambda \geq \theta \pi^2$. The analysis is similar (and far easier) for the case where the open-loop system is asymptotically stable, because in this case we can use the trivial feedback law with $K(1,y)=0$.
\section{Event-triggered  boundary control and main results}\label{Event-triggered-strategySection}

In this section we introduce the event-triggered boundary control and the main results:  the existence of a minimal dwell-time  which is independent of the initial condition,  the well-posedness  and the exponential stability of the closed-loop system  under the event-triggered boundary control.\\
Let us first  define the event-triggered boundary control considered in this paper.  It encloses both a  triggering condition (which determines the time instant at which the controller needs to be sampled/updated) and the backstepping  boundary  feedback  \eqref{control_function_event-triggered_with_K}.  The proposed event-triggering condition  is based on the evolution of the magnitude of the actuation deviation \eqref{deviation_actuation} and the evolution of the $L^2$- norm of the state.   \\

 \begin{Defi}[Definition of the event-triggered boundary control]\label{Definition_event_based_controller}
Let $\beta >0$ and  let  $k(y):= K(1,y)$ with $K$ being the kernel  given in \eqref{solutionKernelexplicit}. The event-triggered boundary control is defined by considering the following  components: \\

\noindent I)   (The event-trigger) The times of the events $t_j\geq 0$  with $t_0=0$  form a finite or countable set of times  which is determined by the following rules for some $j\geq0$: \\
\begin{itemize}
\item [a)] if $ \{ t \in \mathbb{R}_{+}   \vert t >  t_{j}  \wedge     \vert d(t) \vert      >  \beta \Vert  k \Vert \Vert u[t] \Vert +\beta \Vert  k \Vert \Vert u[t_{j}] \Vert     \} = \emptyset$  then the set of the times of the events is $\{t_{0},...,t_{j}\}$.\\ 
\item [b)] if $\{ t \in \mathbb{R}_{+}   \vert t >  t_{j}  \wedge     \vert d(t) \vert      >  \beta \Vert  k \Vert \Vert u[t] \Vert +\beta \Vert  k \Vert \Vert u[t_{j}] \Vert     \} \neq \emptyset$, then the next event time is given by:
%\begin{equation}
%\begin{array}{l} \label{triggering_conditionISS_with_backstepping_original}
%t_{j+1} :=     \inf \{ t \in \mathbb{R}_{+}   \vert t >  t_{j}  \wedge     \vert d(t) \vert      >  \beta \Vert  k \Vert \Vert u[t] \Vert +\beta \Vert  k \Vert \Vert u[t_{j}] \Vert     \}   
%\end{array} 
%\end{equation}  
\begin{equation}\label{triggering_conditionISS_with_backstepping_original}
\begin{split} 
t_{j+1} :=   &  \inf \{ t \in \mathbb{R}_{+}   \vert t >  t_{j}  \wedge     \vert d(t) \vert      >  \beta \Vert  k \Vert \Vert u[t] \Vert \\
& \hskip 2.9cm +\beta \Vert  k \Vert \Vert u[t_{j}] \Vert     \}   
\end{split} 
\end{equation}

\end{itemize}

%
%If  $u[t_{j}] \neq 0$, let the increasing sequence of time instants $(t_{j})$ be defined iteratively by $t_0=0$ , and for all $j \geq 1$,
%\begin{equation}
%\begin{array}{l} \label{triggering_conditionISS_with_backstepping_original}
%   t_{j+1} :=     \inf \{ t \in \mathbb{R}_{+}   \vert t >  t_{j}  \wedge     \vert d(t) \vert      >  \beta \Vert  k \Vert \Vert u[t] \Vert +\beta \Vert  k \Vert \Vert u[t_{j}] \Vert     \}   
%\end{array} 
%\end{equation}
%For the case  that a time $t > t_{j}$  satisfying \eqref{triggering_conditionISS_with_backstepping_original} does not exist,  then $t_{j+1} = +\infty$.
% 
%\noindent  If $u[t_{j}] = 0$,  then   $t_{j+1} = +\infty$.\\

\noindent II) (the control action)  The boundary feedback law,
\begin{equation}\label{operator_controlfunction}
U_d(t) =    \int_{0}^{1} k(y)u(t_{j},y) dy, \quad  \forall t \in [t_{j},t_{j+1})
\end{equation}
\end{Defi}

%\section{Main results}\label{main_results}
%
% Let us first proof that under the event triggered control, there exists a minimal dwell-time. It  follows essentially the same reasoning of \cite{Espitia2018TAC}. As in that work,  the use  of a dynamic triggering condition is instrumental. 
%
\subsection{\textbf{Avoidance of the Zeno phenomena}}
It is worth mentioning that guaranteeing  the existence of a minimal dwell-time between two triggering times  avoids the so-called Zeno phenomena that  means infinite triggering times in a finite-time interval.  It represents  infeasible   practical implementations into digital platforms because it would be required to sample  infinitely fast. 
 % , it would     \footnote{We refer the reader to  
 Before we tackle the result on existence of minimal dwell-time,  let us first introduce the following intermediate result.\\

\begin{Lem}\label{Estimate_of_supNorm}
For the closed-loop system \eqref{eq:sysparabolic}-\eqref{BC_parabolic_PDE_u}, the following estimate holds, for all $t \in [t_{j},t_{j+1}]$, $j\geq 0$:
\begin{equation}\label{upper_bound_supNormofU}
\sup_{t_{j} \leq s \leq t_{j+1}}(\Vert u[s] \Vert) \leq Q \Vert u[t_{j}] \Vert
\end{equation}
 where $Q= e^{p/2(t_{j+1}- t_{j})}(1+ \frac{\sqrt{3}}{3} \Vert k \Vert + \frac{\Vert k \Vert}{\sqrt{p}}) + \frac{\sqrt{3}}{3} \Vert k \Vert$ and $p = -2\theta \pi^2 + 2 \lambda + \frac{1}{3} \lambda^2 $.  % and $q= \theta r^{''}(x) + \lambda r(x)$. 
\end{Lem}
\begin{proof}
 We consider $U_d$ given  by \eqref{operator_controlfunction}  and define 
\begin{equation}\label{change_of_variable_Lemma1}
v(t,x) = u(t,x) - x U_d 
\end{equation}
It is straightforward to verify that $v$ satisfies the following PDE for all $t \in (t_{j},t_{j+1})$, $j\geq 0$,
\begin{eqnarray}\label{eq:sysparabolic-v-fredholmtransform}
%v_t(t,x) & =&  \theta v_{xx}(t,x) + \lambda v(t,x) + \lambda r(x)U_d(t) + r^{''}(x)U_d(t) \\
v_t(t,x) & =&  \theta v_{xx}(t,x) + \lambda v(t,x) +  \lambda x U_d \\
v(t,0)&=&0\\
v(t,1)&=& 0  \label{BC_parabolic_PDE-v-fredholmtransform}
\end{eqnarray} 
Well-posedness issues for \eqref{eq:sysparabolic-v-fredholmtransform}-\eqref{BC_parabolic_PDE-v-fredholmtransform} readily  follows while being a particular case of the PDE considered in  \cite[Lemma 5.2]{Karafyllis_Adaptive-regulation-triggered2019}. Now, 
%and initial condition:
%\begin{equation}\label{IC_parabolic_PDE-v-fredholmtransform}
% w(t_{j},x)=u(t_{j},x) - r(x) U_d(t)
%\end{equation}
by considering the  function $V(t)=\frac{1}{2}\Vert v[t] \Vert^2$   % functional, defined for all $v[t] \in L^2(0,1)$, by  $V(v)=\frac{1}{2}\Vert v[t] \Vert $. 
and taking its time derivative along the solutions of  \eqref{eq:sysparabolic-v-fredholmtransform}- \eqref{BC_parabolic_PDE-v-fredholmtransform} and using the Wirtinger's inequality, we obtain, for $t\in (t_{j}, t_{j+1})$:
\begin{equation*}
%\dot{V} \leq - \theta \pi^2 \Vert v[t] \Vert^2 + \lambda \Vert v[t] \Vert^2 + U_d \int_{0}^{1}(\theta r^{''}(x)+\lambda r(x))v(t,x) dx 
\dot{V} \leq - \theta \pi^2 \Vert v[t] \Vert^2 + \lambda \Vert v[t] \Vert^2 + U_d \int_{0}^{1}(\lambda x)v(t,x) dx 
\end{equation*}  
In addition, using the Young's inequality on the last term along with the Cauchy-Schwarz inequality, we get
\begin{equation*}
\dot{V}(t) \leq - \theta \pi^2 \Vert v[t] \Vert^2 + \lambda \Vert v[t] \Vert^2  + \frac{1}{2}U_d^2 + \frac{1}{6}\lambda^2 \Vert v[t] \Vert^2
\end{equation*}
%where $q= \theta r^{''}(x) + \lambda r(x) $.
 Then, for $t \in (t_{j}, t_{j+1})$:
\begin{equation*}
\dot{V}(t) \leq p V(t) + \frac{1}{2}U_d^2
\end{equation*}
where $p =  -2\theta \pi^2 + 2 \lambda + \frac{1}{3}\lambda^2 $. Using the Comparison principle on an interval $[a,b]$ where $a > t_{j}$ and $b < t_{j+1}$, one gets, for all $t \in [a,b]$:
 \begin{equation*}
 V(t) \leq e^{p(t-a)}(V(a)+ \frac{1}{2 p}U_d^2)
 \end{equation*}
 Due to the continuity of $V(t)$ on $[t_{j},t_{j+1}]$  and the fact that $a,b$ are arbitrary, we can conclude that
 \begin{equation}\label{Estimate_functional_p_tk+1}
  V(t) \leq e^{p(t_{j+1}-t_{j})}\left(V(t_{j})+ \tfrac{1}{2 p}U_d^2\right)
 \end{equation}
 for all $t \in [t_{j},t_{j+1}]$. Using the Cauchy-Schwarz inequality, we have that $\vert U_d \vert \leq \Vert k \Vert \Vert u[t_{j}] \Vert$. Using this fact in  \eqref{Estimate_functional_p_tk+1}, we get, in addition:
 \begin{equation*}
 \Vert v[t] \Vert^2 \leq e^{p(t_{j+1}-t_{j})} \left(\Vert v[t_{j}] \Vert^2 + \tfrac{1}{p} \Vert k \Vert^2 \Vert u[t_{j}] \Vert^2 \right)
\end{equation*}  
Using the above estimate in conjunction  with \eqref{change_of_variable_Lemma1} and the triangle inequalities, we obtain the following inequalities:  
\begin{equation*}
\begin{split}
\Vert u[t] \Vert \leq \Vert v[t]\Vert +  \tfrac{\sqrt{3}}{3} \vert U_d\vert \\
\Vert v[t_{j}] \Vert \leq \Vert u[t_{j}]\Vert +  \tfrac{\sqrt{3}}{3} \vert U_d\vert
\end{split}
\end{equation*}
together with $\vert U_d \vert \leq \Vert k \Vert \Vert u[t_{j}] \Vert$, we finally obtain, for all $t \in [t_{j}, t_{j+1}]$,
\begin{equation*}
\sup_{t_{j} \leq s \leq t_{j+1}}(\Vert u[s] \Vert) \leq Q \Vert u[t_{j}] \Vert
\end{equation*}
with $Q= e^{p/2(t_{j+1}- t_{j})}(1+  \tfrac{\sqrt{3}}{3} \Vert k \Vert + \frac{\Vert k \Vert}{\sqrt{p}}) +  \tfrac{\sqrt{3}}{3} \Vert k \Vert$.  This concludes the proof.
\end{proof}

\begin{Theo}\label{theo:minimal_dwellt_time}
Under the event-triggered boundary control \eqref{triggering_conditionISS_with_backstepping_original}-\eqref{operator_controlfunction},  there exists a minimal dwell-time between two triggering times, i.e.  there exists a constant $\tau >0$ (independent of the initial condition $u_0$) such that $t_{j+1} - t_{j} \geq \tau $,  for all $j \geq 0$.
\end{Theo}

\begin{proof}

Define $g \in C^2([0,1])$  by the following equation:
 \begin{equation}\label{definition_g_modaldecopsition}
 g(x):= \sum_{n=1}^{N}k_n \phi_{n}(x)
\end{equation}  
where $N\geq 1$ is an integer, $k_n := \int_{0}^{1}k(y)\phi_n(y)dy$, $k(y)=K(1,y)$ with $K$ satisfying \eqref{solutionKernelexplicit} and  $\phi_n(x) =\sqrt{2}\sin(n \pi x) ,n=1,2...$  are the eigenfunctions  of the Sturm-Liouville operator $A: D \rightarrow L^2(0,1)$ defined by
\begin{equation*}
(Af)(x)= - \theta \frac{d^2 f}{dx^2}(x) - \lambda f(x)
\end{equation*}
for all $f \in D$ and $x \in (0,1)$ and $D \subset H^2([0,1])$ is the set of functions $f : [0,1] \rightarrow \mathbb{R} $ for which $f(0)= f(1)=0 $. \\

Let us also define
%As in \cite{KARAFYLLIS2018226}, let us define:
\begin{equation}\label{tilde_deviation}
\tilde{d}(t)= \int_{0}^{1} g(y)\left(u(t_{j},y) - u(t,y)\right)dy
\end{equation}
for $t \in [t_{j},t_{j+1})$, for $j\geq 0$ and $g$ given  by \eqref{definition_g_modaldecopsition}. Taking the time derivative of  $\tilde{d}(t)$ along the solutions of \eqref{eq:sysparabolic}-\eqref{BC_parabolic_PDE_u}  yields, for all $t \in [t_{j},t_{j+1})$:
\begin{equation*}
\begin{split}
\dot{\tilde{d}}(t) =& \theta \left( \frac{dg}{dx}(1)u(t,1) - g(1)\frac{\partial u}{\partial x}(t,1) \right) \\
& + \theta \left(g(0)\frac{\partial u}{\partial x}(t,0) - \frac{dg}{dx}(0)u(t,0) \right) \\
&+ \int_{0}^{1}(Ag)(y)u(t,y)dy
\end{split}
\end{equation*}
Note that  $g(1)\frac{\partial u}{\partial x}(t,1) = 0 $  by virtue of the function  $g$ evaluated at $x=1$ as $\phi_n(1)=0$. 
In addition, by  the eigenvalue problem $A\phi_n = \lambda_n \phi_n$ where $\lambda_n = n^2\pi^2\theta - \lambda $  are real  eigenvalues %(having the property    $\lambda_1 < \lambda_2 < ... < \lambda_n< ..$ with $\lim_{n \rightarrow \infty} (\lambda_n) = + \infty$)  
and using the boundary conditions   \eqref{BC_parabolic_PDE_u}, we get 
\begin{equation*}
\begin{split}
\dot{\tilde{d}}(t) =& \theta \int_{0}^{1}k(y)u(t_{j},y)dy \sum_{n=1}^{N}k_n \frac{d \phi_n}{dx}(1) \\
& + \sum_{n=1}^{N}k_n \lambda_n \int_{0}^{1}\phi_n(y)u(t,y)dy
\end{split}
\end{equation*}
Using the Cauchy-Schwarz inequality and $\Vert \phi_n \Vert=1$ for $n=1,2,...$ the following estimate holds for $t \in (t_{j},t_{j+1})$, $j\geq 0$:
 \begin{equation}\label{estimate_tilded2}
\vert \dot{\tilde{d}}(t) \vert   \leq  \theta \Vert k \Vert \Vert u[t_{j}] \Vert F_N + \Vert u[t] \Vert G_{N}
\end{equation}
where   $F_N := \sum_{n=1}^{N}\Big\vert k_n \frac{d \phi_n}{dx}(1)  \Big \vert $ and  $G_N := \sum_{n=1}^{N} \vert k_n \lambda_n \vert $.
Therefore, we obtain from \eqref{estimate_tilded2} and the fact that $\tilde{d}(t_{j})=0$,  the following estimate:
 \begin{equation}\label{estimate_tilded3}
\vert \tilde{d}(t) \vert \leq (t -t_{j})\theta \Vert k \Vert \Vert u[t_{j}] \Vert F_N + (t-t_{j})\sup_{t_{j} \leq s \leq t}(\Vert u[s] \Vert) G_N
\end{equation}
Note that  from \eqref{deviation_actuation},\eqref{operator_controlfunction} and \eqref{tilde_deviation}, the deviation $d(t)$ can be expressed as follows:
\begin{equation}\label{d(t)_in_termsof_Tilde_d}
d(t) = \tilde{d}(t) + \int_{0}^{1}(k(y)-g(y))(u(t_{j},y) - u(t,y))dy 
\end{equation} 
Hence,  combining \eqref{estimate_tilded3} and \eqref{d(t)_in_termsof_Tilde_d} we can obtain an estimate of $d$ as follows:  
\begin{equation}\label{esdtimateod_d_withnormk-g}
\begin{split}
\vert d(t) \vert & \leq (t -t_{j})\theta \Vert k \Vert \Vert u[t_{j}] \Vert F_N + (t-t_{j})\sup_{t_{j} \leq s \leq t}(\Vert u[s] \Vert) G_N  \\
& \hskip 2 cm+ \Vert k-g \Vert\Vert u[t_{j}] \Vert + \Vert k-g \Vert\Vert u[t] \Vert
\end{split}
\end{equation}
Using \eqref{esdtimateod_d_withnormk-g} and assuming that an event is triggered  at $t=t_{j+1}$, we have
\begin{equation}\label{estimate_d_at_tk+1}
\begin{split}
\vert d(t_{j+1}) \vert & \leq (t_{j+1} -t_{j})\theta \Vert k \Vert \Vert u[t_{j}] \Vert F_N \\
& \hskip 2 cm + (t_{j+1}-t_{j})\sup_{t_{j} \leq s \leq t_{j+1}}(\Vert u[s] \Vert) G_N  \\
& \hskip 2 cm+ \Vert k-g \Vert\Vert u[t_{j}] \Vert + \Vert k-g \Vert\Vert u[t_{j+1}] \Vert
\end{split}
\end{equation}
and,  if $u[t_{j}] \neq 0$, by  Definition \ref{Definition_event_based_controller},  we have that, at $t=t_{j+1}$
\begin{equation}\label{d_at_tk+1}
\vert d(t_{j+1}) \vert \geq  \beta \Vert k \Vert \Vert u[t_{j}] \Vert + \beta \Vert k \Vert \Vert u[t_{j+1}] \Vert 	
\end{equation}
Combining \eqref{estimate_d_at_tk+1} and \eqref{d_at_tk+1}, we get
\begin{equation*}
\begin{split}
&\beta\Vert k \Vert \Vert u[t_{j}] \Vert  + \beta \Vert k \Vert \Vert u[t_{j+1}] \Vert   \\
&  \leq (t_{j+1} -t_{j})\theta \Vert k \Vert F_N \Vert u[t_{j}] \Vert  + (t_{j+1}-t_{j})G_N  \sup_{t_{j} \leq s \leq t_{j+1}}(\Vert u[s] \Vert)  \\
&  \hskip 2cm + \Vert k-g \Vert\Vert u[t_{j}] \Vert + \Vert k-g \Vert\Vert u[t_{j+1}] \Vert
\end{split}
\end{equation*}
therefore,
\begin{equation*}
\begin{split}
 & \left(\beta\Vert k \Vert  -  \Vert k-g \Vert  \right) \Vert u[t_{j+1}] \Vert +  \left( \beta \Vert k \Vert  -  \Vert k-g \Vert  \right) \Vert u[t_{j}] \Vert   \\
  &  \leq (t_{j+1} -t_{j})\theta \Vert k \Vert F_N \Vert u[t_{j}] \Vert     +   (t_{j+1}-t_{j})G_N\sup_{t_{j} \leq s \leq t_{j+1}}(\Vert u[s] \Vert)
\end{split}
\end{equation*}
We select $N \geq 1$ in \eqref{definition_g_modaldecopsition} sufficiently large so  that $\Vert k -g \Vert < \beta\Vert k \Vert$. 
 Notice that we can always find $N$ sufficiently large so that the condition $\Vert k -g \Vert < \beta\Vert k \Vert$, since $g$ is simply the $N$-mode truncation of $k$ (which implies that $\Vert k-g\Vert$ tends to zero as $N$ tends to infinity).
In addition, using the fact that $\Vert u[t_{j+1}] \Vert \geq 0$ and  by \eqref{upper_bound_supNormofU} in Lemma \ref{Estimate_of_supNorm}, we obtain the following estimate:
\begin{equation*}
\begin{split}
& \left( \beta\Vert k \Vert  -  \Vert k-g \Vert  \right) \Vert u[t_{j}] \Vert    \\
 &  \leq (t_{j+1} -t_{j})\theta \Vert k \Vert F_N \Vert u[t_{j}] \Vert + (t_{j+1}-t_{j})G_N Q \Vert u[t_{j}] \Vert    \\
\end{split}
\end{equation*} 
where $Q = e^{p/2(t_{j+1}- t_{j})}(1+  \tfrac{\sqrt{3}}{3} \Vert k \Vert + \frac{\Vert k \Vert}{\sqrt{p}}) +  \tfrac{\sqrt{3}}{3} \Vert k \Vert$ and $p = -2\theta \pi^2 + 2 \lambda + \frac{1}{3}\lambda^2 $. %\footnote{{\color{blue} $r(x)$ can be taken here as  $r(x)= x - g(x)$. }} \\
Denoting 
\begin{itemize}
\item $a_0:= \beta\Vert k \Vert  -  \Vert k-g \Vert $
\item $a_1:= \theta \Vert k \Vert F_N + G_N  \tfrac{\sqrt{3}}{3} \Vert k \Vert$
\item $a_2:= G_N (1 +  \tfrac{\sqrt{3}}{3} \Vert k \Vert + \tfrac{\Vert k \Vert}{\sqrt{p}})$
\end{itemize}
we obtain an inequality of the form:
\begin{equation}\label{inequality_for_dwell_time_LambertW}
a_0 \leq a_1 (t_{j+1} - t_{j}) + a_2 (t_{j+1} - t_{j}) e^{p/2 (t_{j+1} - t_{j})}
\end{equation}
from which we aim at finding a lower bound for $(t_{j+1} - t_{j})$. Note that the right hand side of \eqref{inequality_for_dwell_time_LambertW} is a $\mathcal{K}_{\infty}$ function of $(t_{j+1} - t_{j})$. Let us denote it as   $\alpha(s):= a_1 s + a_2 s e^{p/2 s}  $ with $s= (t_{j+1} - t_{j})$. The solution of inequality \eqref{inequality_for_dwell_time_LambertW} is then $ s \geq \alpha^{-1}(a_0)$ where $\alpha^{-1}$ is the inverse of $\alpha$. Since $a_0$ is strictly positive, then there exists $\tau >0$ such that $s= (t_{j+1} - t_{j})\geq \tau > 0$. \\
If $u[t_{j}] = 0$, then Lemma \ref{Estimate_of_supNorm} guarantees that $u[t]$ remains zero. In this case, by Definition \ref{Definition_event_based_controller},  one would not need to trigger anymore and thus Zeno phenomenon is immediately   excluded.  It concludes the proof.
\end{proof} 

\begin{remark}[Explicit dwell-time]
We can upper bound the right-hand side of \eqref{inequality_for_dwell_time_LambertW} such that:

\begin{equation}\label{inequality_for_dwell_time_LambertW2}
a_0 \leq (a_1 +a_2 ) (t_{j+1} - t_{j})  e^{p/2 (t_{j+1} - t_{j})}
\end{equation}
which turns out to be more conservative (thus,   one can expect  solutions  of \eqref{inequality_for_dwell_time_LambertW} to take smaller values). Furthermore, by 
rewriting  \eqref{inequality_for_dwell_time_LambertW2} yields,

\begin{equation}\label{inequality_for_dwell_time_LambertW3}
 \frac{p a_0}{2( a_1  + a_2  )} \leq \tfrac{p}{2}(t_{j+1} - t_{j}) e^{p/2 (t_{j+1} - t_{j})}
\end{equation}
so that the right-hand side corresponds to a transcendental function whose solution can be found using the so-called \textit{ Lambert W  function} \footnote{To the best of our knowledge, in control theory, Lambert W functions have been used  within the framework of time-delay systems (see e.g. \cite{Yi2010}).} (see e.g. \cite{Corless1996} for more details).\\
 Hence,  we have
\begin{equation}\label{inequality_for_dewell_time_LambertW4}
(t_{j+1} - t_{j}) \geq \frac{2}{p}W\left(\frac{p a_0}{2(a_1+a_2)}\right)
\end{equation} 
Note that the argument $\tfrac{p a_0}{2(a_1+a_2)}$ of the Lambert W function  is strictly positive yielding  $W(\cdot)$ to take a strictly positive value. We denote then $\tau  := \frac{2}{p} W \left( \tfrac{p a_0}{2(a_1+a_2)} \right) > 0$ being $\tau$ the minimal dwell-time between two triggering times,  i.e. $t_{j+1} - t_{j} \geq \tau$ for all $j\geq 0$. 
 \end{remark}
\vskip 0.5cm 
\begin{remark}\label{Remark_aboutperiodicscheeme}
It is worth remarking that if a periodic sampling scheme - where the the control value is updated periodically on a sampled-and-hold manner - is intended to be applied to  stabilize the system \eqref{eq:sysparabolic}-\eqref{IC_parabolic_PDE_u}, then   a period can be chosen according to \cite{KARAFYLLIS2018226}. However,  an alternative   way of choosing a suitable period while meeting the theoretical guarantees,  is precisely  by using the minimal dwell-time $\tau$  that was obtained from  \eqref{inequality_for_dwell_time_LambertW} and its explicit form by  the Lambert W function as in \eqref{inequality_for_dewell_time_LambertW4}.  Unfortunately,  as one may expect, such a dwell-time may be very small and conservative (similar to the sampling period obtained  in \cite{KARAFYLLIS2018226}  since the derivation was done using small gain arguments as well). This issue, however,   supports  the main motivation highlighted throughout the paper: stabilize on events only when is required and in an more efficient way. Event-triggered control offers  advantages with respect period schemes as it  reduces execution times and while meeting  theoretical guarantees. 
\end{remark}

Since  there is a minimal dwell-time (which is uniform and does not depend on  either  the initial condition or on the state of the system), no Zeno solution can appear.  Consequently,  the following result on  the existence of solutions in  of the system system \eqref{eq:sysparabolic}-\eqref{IC_parabolic_PDE_u} with  \eqref{triggering_conditionISS_with_backstepping_original}-\eqref{operator_controlfunction}, holds for all $ t \in \mathbb{R}_{+}$.\\

\begin{corollary}\label{existence_and_continuity_solutions_on_R}
For every $u_{0} \in L^2(0,1)$, there exist a unique solution $u \in  \mathcal{C}^{0}(\mathbb{R}_{+}; L^{2}(0,1))$ to the 
system \eqref{eq:sysparabolic}-\eqref{IC_parabolic_PDE_u}, \eqref{triggering_conditionISS_with_backstepping_original},\eqref{operator_controlfunction}  satisfying  $u \in C^{1}(I \times [0,1])$, $u[t] \in C^2([0,1])$ for all $t > 0 $ where $I = \mathbb{R}_{+} \backslash \{ t_{j} \geq 0, k= 0,1,2,...\} $
\end{corollary}
\begin{proof}
It is an immediate  consequence of Proposition \ref{existence_and_continuity_solutions} and Theorem \ref{theo:minimal_dwellt_time}. Indeed, the solution is constructed  (by the step method) iteratively between successive triggering times.  
\end{proof}

\subsection{\textbf{Stability result}}

In this subsection,  we  are going to follow small-gain arguments and seek for an Input-to-State stability property with respect to the deviation $d(t)$. \\

\begin{Lem}[ISS of the target system]\label{ISS_target_system}

The target system \eqref{Target_system_w}-\eqref{IC_Target_system_w} is ISS with respect to $d(t)$; more precisely, the following estimate holds:
\begin{equation}\label{ISS_estimate_targetsystem}
\Vert w[t] \Vert \leq G e^{- \sigma t} \Vert w[0] \Vert + \gamma  \sup_{0\leq s \leq t}\left(\vert d(s) \vert e^{-\sigma(t-s)}\right)
\end{equation}
where   $\sigma \in (0, \mu_1)$ with  $\mu_1=\pi^2\theta +c$,  $G := \sqrt{(1+b^{-1})}$, for arbitrary $b >0$ and  the gain $\gamma$  is given as follows:

\begin{equation}\label{small_gain_parameter}
\gamma:= \sqrt{(1+b)} \begin{cases}
\left( \frac{\pi^2\theta + c}{\pi^2\theta + c - \sigma} \right)       \frac{\left( \sinh\left(2\frac{\sqrt{c}}{\sqrt{\theta}}\right) - 2 \frac{\sqrt{c}}{\sqrt{\theta}} \right)}{2\sinh\left(\frac{\sqrt{c}}{\sqrt{\theta}}\right)\left(\frac{c}{\theta}\right)^{1/4}},  \quad \textit{if} \quad c  \neq 0 \\
  \frac{1}{\sqrt{3}}\left( \frac{\pi^2\theta }{\pi^2\theta  - \sigma} \right),   \quad \textit{if} \quad c  = 0
\end{cases}
\end{equation}

  \end{Lem}

\begin{proof}
See  \cite[Appendix]{KARAFYLLIS2018226}.
\end{proof}

\begin{Theo}\label{main_theorem_EBC_backstepping}

Let $\tilde{L} := 1+ \left(\int_{0}^{1} \left( \int_{0}^{x}\vert L(x,y) \vert^2 dy \right) dx \right)^{1/2} $ with $L$ satisfying \eqref{solutionKernelexplicit_inverse} and  $k(y)=K(1,y)$ with $K$ satisfying \eqref{solutionKernelexplicit}.  Let $\gamma$ be given as in  Lemma \ref{ISS_target_system}.   Let $\beta>0 $ be as  in \eqref{triggering_conditionISS_with_backstepping_original}. If   the following condition is fulfilled,
\begin{equation}\label{coondition_for_small_gain}
 \Phi_e:= 2\beta \gamma \Vert k \Vert \tilde{L}  < 1
\end{equation} 
then,  the  closed-loop system  \eqref{eq:sysparabolic}-\eqref{IC_parabolic_PDE_u}  with event-triggered boundary control    \eqref{triggering_conditionISS_with_backstepping_original}-\eqref{operator_controlfunction} has a unique solution and is globally exponentially stable; i.e. there exist $M, \sigma >0$ such that  for every $u_0 \in L^2(0,1)$ the unique mapping $u \in  \mathcal{C}^{0}(\mathbb{R}_{+}; L^{2}(0,1))$ satisfying  $u \in C^{1}(I \times [0,1])$, $u[t] \in C^2([0,1])$ for all $t > 0 $ where $I = \mathbb{R}_{+} \backslash \{ t_{j} \geq 0, k= 0,1,2,...\} $
satisfies:
\begin{equation}\label{GES}
\Vert u[t] \Vert \leq M e^{-\sigma t} \Vert u[0] \Vert, \quad \text{for all} \quad t\geq 0
\end{equation}
 \end{Theo}
\begin{proof}
By Corollary \ref{existence_and_continuity_solutions_on_R}, the existence and uniqueness of a solution to the system \eqref{eq:sysparabolic}-\eqref{IC_parabolic_PDE_u}  with event-triggered boundary control    \eqref{triggering_conditionISS_with_backstepping_original}-\eqref{operator_controlfunction} hold. Let us show that the system is globally exponential stable in  the $L^2$-norm sense. \\

 It follows from \eqref{triggering_conditionISS_with_backstepping_original}  that the following inequality holds for all $t \in [t_{j},t_{j+1})$:
\begin{equation}\label{estimateofd2}
\vert d(t) \vert \leq \beta \Vert k \Vert \Vert u[t_{j}] \Vert + \beta \Vert k \Vert \Vert u[t] \Vert
\end{equation}
Therefore, inequality \eqref{estimateofd2} implies the following inequality for all $t\geq 0$:

\begin{equation}\label{estimateofsup-d}
\sup_{0\leq s \leq t} \left( \vert d(s)\vert  e^{\sigma s}  \right) \leq 2 \beta \Vert k \Vert \sup_{0 \leq s \leq t} \left( \Vert u[s] \Vert e^{\sigma s} \right) 
\end{equation}
On the other hand, by Lemma \ref{ISS_target_system}, we have
\begin{equation}\label{ISS_target-system-in_main_theorem}
\Vert w[t] \Vert e^{\sigma t} \leq G \Vert w[0] \Vert + \gamma \sup_{0\leq s \leq t}\left(\vert d(s) \vert e^{\sigma s}\right)
\end{equation}
 The following estimate is a consequence of \eqref{ISS_target-system-in_main_theorem}: 
\begin{equation}\label{estimateofsup-d2}
\sup_{0\leq s \leq t}\left( \Vert w[s] \Vert e^{\sigma s} \right) \leq G \Vert w[0] \Vert + \gamma \sup_{0\leq s \leq t}\left(\vert d(s) \vert e^{\sigma s}\right)
\end{equation}
Hence, combining \eqref{estimateofsup-d} with  \eqref{estimateofsup-d2}, we obtain %using the fact $\Vert u[t] \Vert \leq \tilde{L} \Vert w[t] \Vert$ and the definition  $\Vert w \Vert_{[0,t]}:=\sup_{0\leq s \leq t}\left( \Vert w[s] \Vert e^{\sigma s} \right)$   we obtain
\begin{equation*}
\sup_{0\leq s \leq t}\left( \Vert w[s] \Vert e^{\sigma s} \right) \leq G \Vert w[0] \Vert + 2 \beta  \gamma \Vert k \Vert \sup_{0 \leq s \leq t} \left( \Vert u[s] \Vert e^{\sigma s} \right) 
\end{equation*}
and using the fact $\Vert u[t] \Vert \leq \tilde{L} \Vert w[t] \Vert$, we get
\begin{equation}\label{estimateofsup-d3}
\sup_{0\leq s \leq t}\left( \Vert w[s] \Vert e^{\sigma s} \right) \leq G \Vert w[0] \Vert + \Phi_e \sup_{0 \leq s \leq t} \left( \Vert w[s] \Vert e^{\sigma s} \right) 
\end{equation}
%\begin{equation}
%(1- \Phi_e)\Vert w \Vert_{[0,t]} \leq G \Vert w[0] \Vert
%\end{equation}
where
\begin{equation}\label{Phi_event2}
\Phi_e:= 2 \beta  \gamma \Vert k \Vert \tilde{L}
\end{equation}
 Notice that, by virtue of \eqref{coondition_for_small_gain}, it holds that 
 $ \Phi_e < 1$. Thereby,  using the estimate of the backstepping transformation, i.e. $\Vert w[t] \Vert \leq \tilde{K} \Vert u[t] \Vert$ with $\tilde{K} := 1+ \left(\int_{0}^{1} \left( \int_{0}^{x}\vert K(x,y) \vert^2 dy \right) dx \right)^{1/2} $ and $K$ satisfying \eqref{solutionKernelexplicit},   we obtain  from \eqref{estimateofsup-d3} and \eqref{Phi_event2}  the following estimate for the solution  to the closed-loop system \eqref{eq:sysparabolic}-\eqref{IC_parabolic_PDE_u}  with event-triggered control    \eqref{triggering_conditionISS_with_backstepping_original}-\eqref{operator_controlfunction}:
\begin{equation*}
\sup_{0 \leq s \leq t} \left( \Vert u[s] \Vert e^{\sigma s} \right)  \leq G(1-\Phi_e)^{-1} \tilde{K}\tilde{L} \Vert u[0] \Vert
\end{equation*}
which leads to \eqref{GES}:
\begin{equation*}
\Vert  u[t] \Vert \leq M e^{-\sigma t}\Vert u[0] \Vert
\end{equation*}
with $M:= G(1-\Phi_e)^{-1} \tilde{K}\tilde{L}$. It concludes the proof.

\end{proof}

%{\color{red}
%NOte that The $(\tilde{\alpha},\tilde{beta})$-system and the homogeneous part of the $(\hat{\alpha},\hat{\tilde{beta}})$-system (i.e. without $\tilde{\alpha(t,0)}$ term) are exponentially..
%}
%\begin{nnremark}\label{on_choiceparameters}
%{\color{blue}
%Remark on the choice of parameters... to be completed
%%Sampling faster.  Violation of the small gain condition $\Phi_e$. 
%}
%
%%The sampling speed is adjust thorugh the parameters $\alpha$ and $\beta$ in \eqref{} provided that \eqref{} is satisfied. ONe can expect that larger  results in larger inter-execution times. Some observations through numerical simulations are provided in the sequel. However, an optimal choice of parameters are not investigated in thiw paper and is left for future works.   Optimal choice of parameters while meeting the theoretical guaranteeing and obrtaiing performance.
%\end{nnremark}
% 

\section{Numerical simulations}\label{numerical_simulation}

We consider the reaction-diffusion system  with $\theta = c= 1$, $\lambda=\pi^2$ and %Kernels of the Backstepping transformation are given by
%\begin{equation*}
%K(x,y)=-y \tfrac{\lambda + c}{\theta} \frac{I_1\left(\sqrt{\tfrac{\lambda + c}{\theta}(x^2-y^2)   } \right)}{\sqrt{\tfrac{\lambda + c}{\theta}(x^2-y^2)   }},  \quad  L(x,y)=-y \tfrac{\lambda + c}{\theta} \frac{J_1\left(\sqrt{\tfrac{\lambda + c}{\theta}(x^2-y^2)   } \right)}{\sqrt{\tfrac{\lambda + c}{\theta}(x^2-y^2)   }}
%\end{equation*}
initial condition $u_0(x) = \sum_{n=1}^3 \frac{\sqrt{2}}{n}\sin(n \pi x) + 3(x^2 - x^3)$, $x \in [0,1]$. For numerical simulations, the state of the system has been discretized by divided differences on a uniform grid with the step $h=0.01$ for the space variable. The discretization with respect to time was done using the implicit Euler scheme with step size $\Delta t=h^2 $.

We stabilize the system on events under  the event-triggered boundary control \eqref{triggering_conditionISS_with_backstepping_original}-\eqref{operator_controlfunction}  where   the parameter $\beta$ is selected such that condition  \eqref{coondition_for_small_gain} in Theorem \ref{main_theorem_EBC_backstepping}  is verified. In addition, $\tilde{L} =  1.8407$, $\Vert k \Vert = 5.61$ and $\gamma= 0.574 $ which is computed according to the information provided in Lemma \ref{ISS_target_system}. Therefore,  
  two cases are pointing out:  we  choose e.g. $\beta = 0.07$ and $\beta = 0.02$ yielding $\Phi_e = 0.83 <1$ and $\Phi_e = 0.23 <1$, respectively. In the former case,   $12$ events (updating times of the control) are obtained whereas in the later case, $47$ events are obtained. 
Figure \ref{component_solution} shows   the numerical solution of the closed-loop system \eqref{eq:sysparabolic}-\eqref{IC_parabolic_PDE_u}  with event-triggered control    \eqref{triggering_conditionISS_with_backstepping_original}-\eqref{operator_controlfunction} (on the left   $\beta=0.07$ and on the right when  $\beta=0.02$,  which results in slow and fast sampling, respectively).    The time-evolution of control functions  under the  event-triggered case is shown in  Figure \ref{control_functions} (orange line with black circle marker  for slow sampling   and  blue line with red circle marker for fast sampling).
\begin{figure*}[t]
%%\captionsetup[subfigure]{labelformat=empty}
\centering{
\subfigure{\includegraphics[width=1\columnwidth]{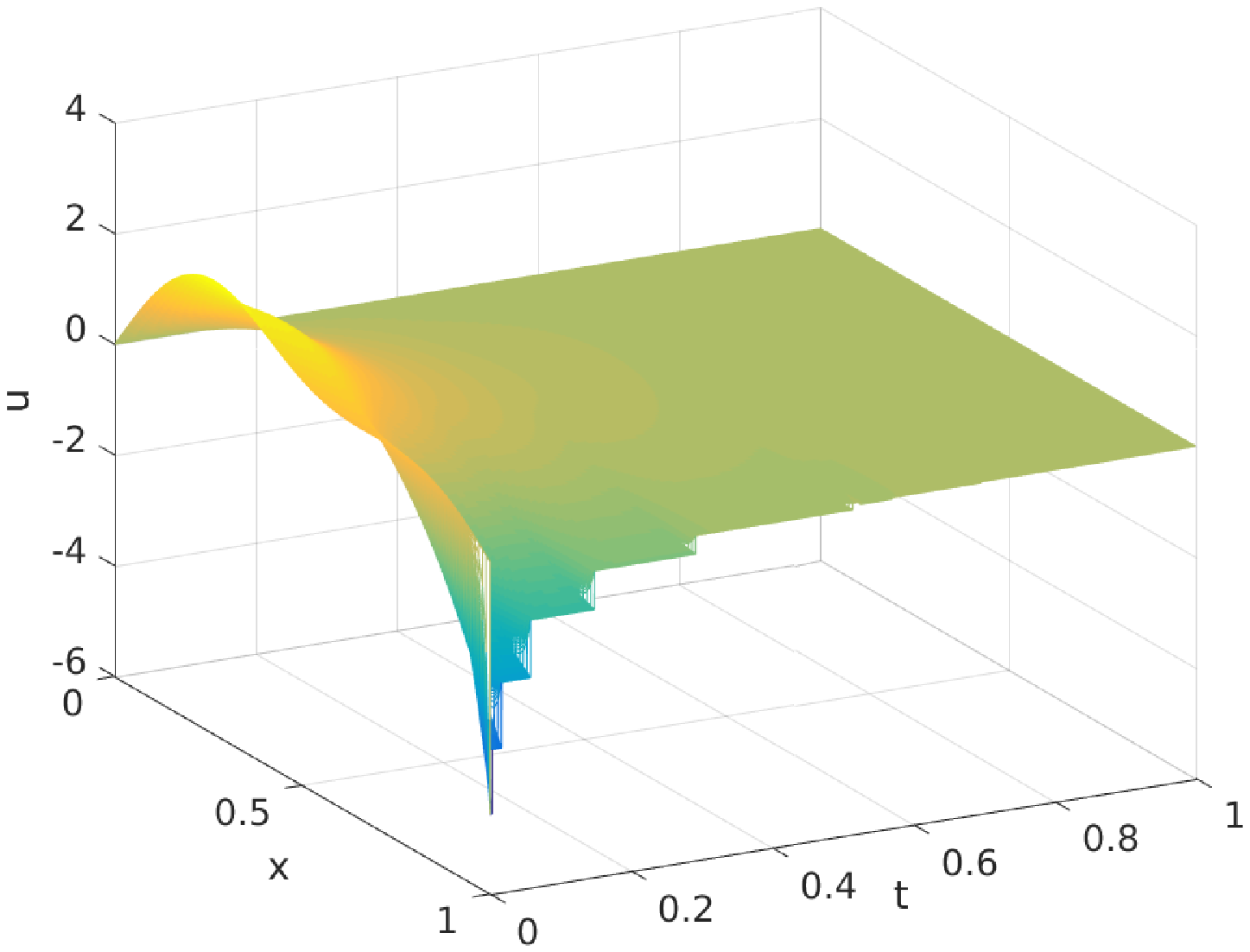} }
\subfigure{\includegraphics[width=1\columnwidth]{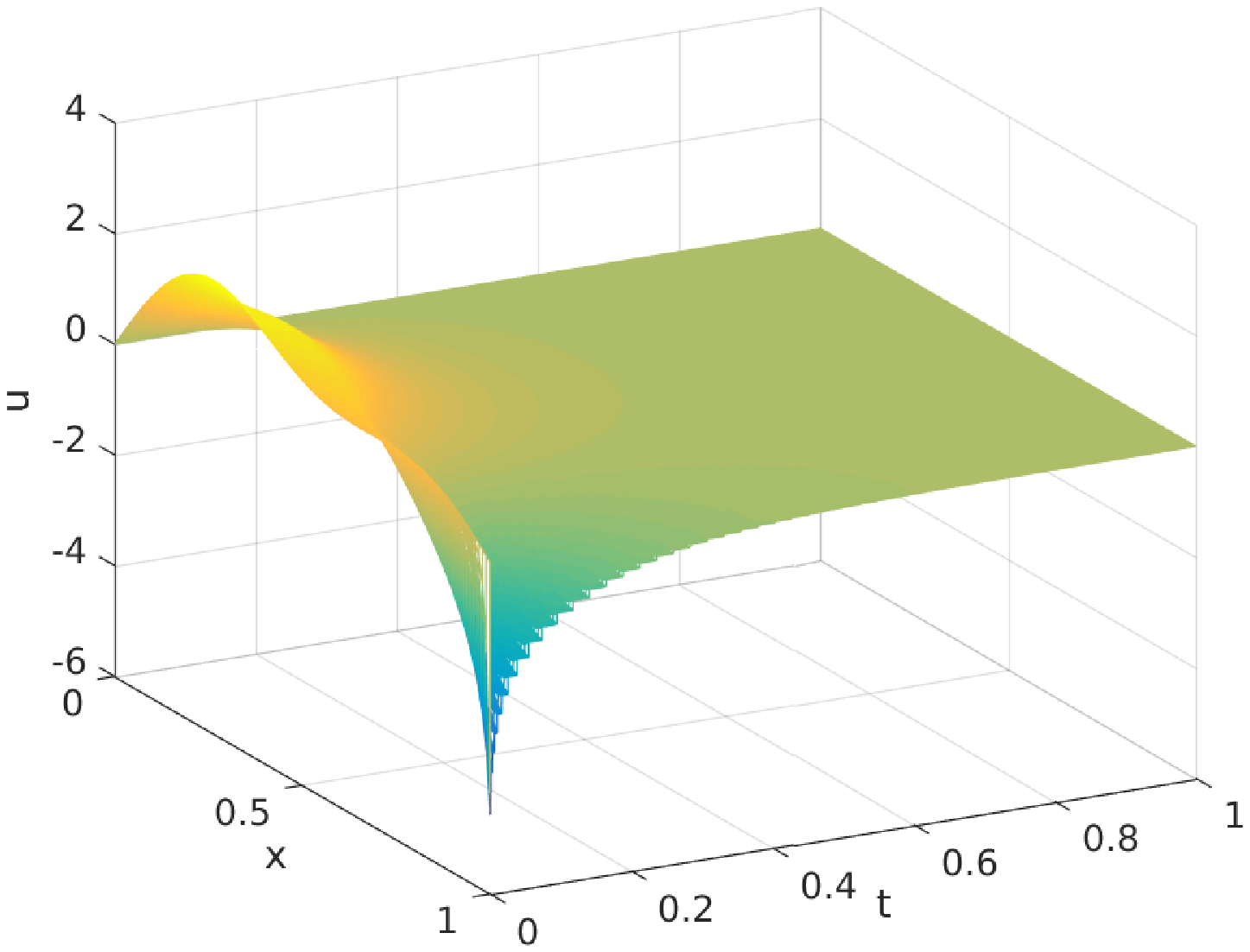} }%scaling font 120% in matlab
\caption{Numerical solutions of the closed-loop system \eqref{eq:sysparabolic}-\eqref{IC_parabolic_PDE_u} with $\theta = c= 1$, $\lambda=\pi^2$, initial condition $u_0(x) = \sum_{n=1}^3 \frac{\sqrt{2}}{n}\sin(n \pi x) + 3(x^2 - x^3)$, $x \in [0,1]$ and  under the event-triggered control    \eqref{triggering_conditionISS_with_backstepping_original}-\eqref{operator_controlfunction}. With  $\beta =0.07$ in \eqref{triggering_conditionISS_with_backstepping_original} the control updating is slower (closed-loop solution depicted on the left). With $\beta =0.02$ in \eqref{triggering_conditionISS_with_backstepping_original}, the control updating is faster (closed-loop solution depicted on the right).} 
\label{component_solution}
}
\end{figure*}
\begin{figure}[t]  
\centering
\subfigure{\includegraphics[width=1 \columnwidth]{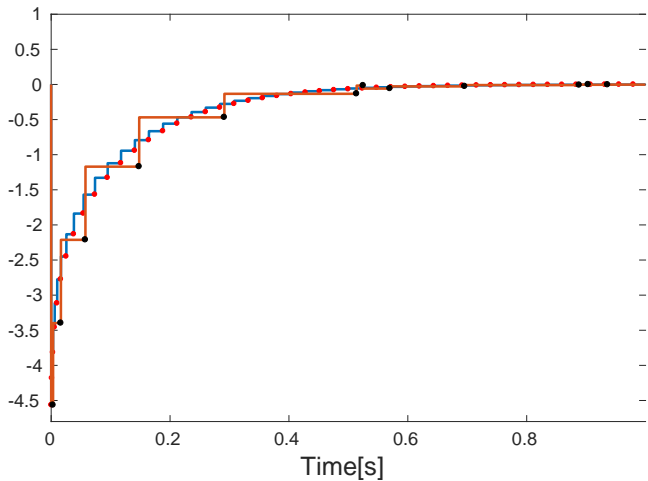}     }
\caption{Time-evolution of the event-triggered boundary control \eqref{triggering_conditionISS_with_backstepping_original}-\eqref{operator_controlfunction} (orange line with black circle marker  for slow control updating, i.e. $\beta=0.07$ in  \eqref{triggering_conditionISS_with_backstepping_original}  and  blue line with red circle marker for fast control updating, i.e. $\beta=0.02$ in  \eqref{triggering_conditionISS_with_backstepping_original}).}
\label{control_functions}
\end{figure}
In addition, Figure \ref{trajectories_triggering_condition} shows the  time evolution of the functions appearing in the
triggering condition \eqref{triggering_conditionISS_with_backstepping_original} (on the left with  $\beta=0.07$ and on the right with  $\beta=0.02$). Once the trajectory $ \vert d \vert$ reaches the trajectory $\beta \Vert  k \Vert \Vert u[t] \Vert +\beta \Vert  k \Vert \Vert u[t_{j}] \Vert $, an event is generated, the control value is updated
and $d$ is reset to zero.  It can be observed that   the lower  $\beta$ is, the faster the sampling and control updating which in turn implies smaller inter-executions times. This case turns out to be more conservative and the control function gets closer to the one in continuous case or  even when considering a periodic scheme with a very small period. As a matter of fact,    it is worth remarking that a sampling period can be computed from \cite[Section 3.3]{KARAFYLLIS2018226}. Indeed, for  the reaction-diffusion system  with a boundary control whose actuation is done in a sampled-and-hold fashion,  such a period would be $T = 9.96 \times 10^{-7}$.  
Notice that this is very small (even  smaller than the time step discretization for the current simulations); consequently   the periodic scheme turns out  not  be implementable. This is one of the reasons why  event-triggered boundary control offers advantages with respect to periodic schemes. In our framework, the control value is updated aperiodically and only when needed. 
 
\begin{figure*}[t]
%%\captionsetup[subfigure]{labelformat=empty}
\centering{
\subfigure{\includegraphics[width=1\columnwidth]{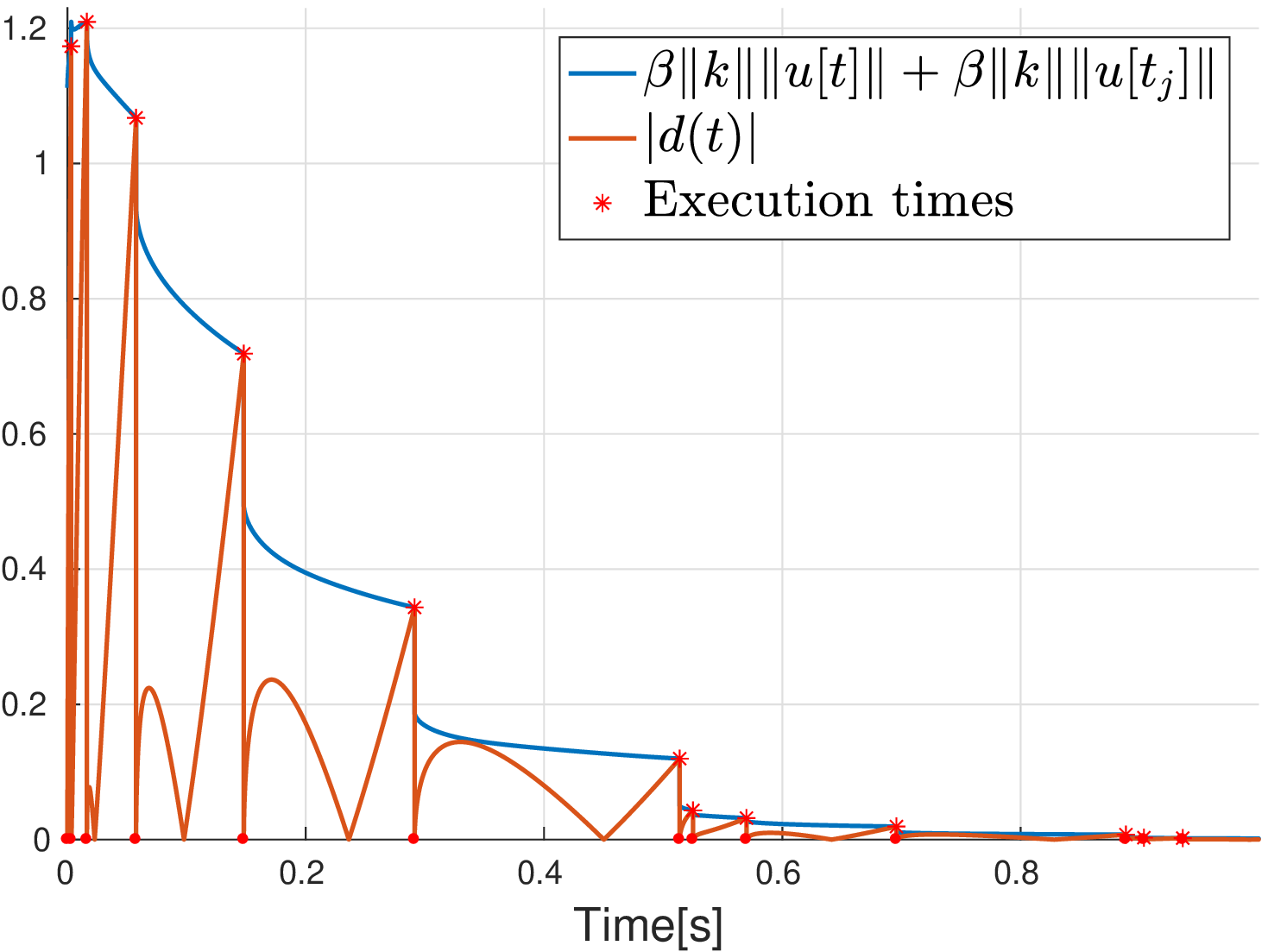} }
\subfigure{\includegraphics[width=1\columnwidth]{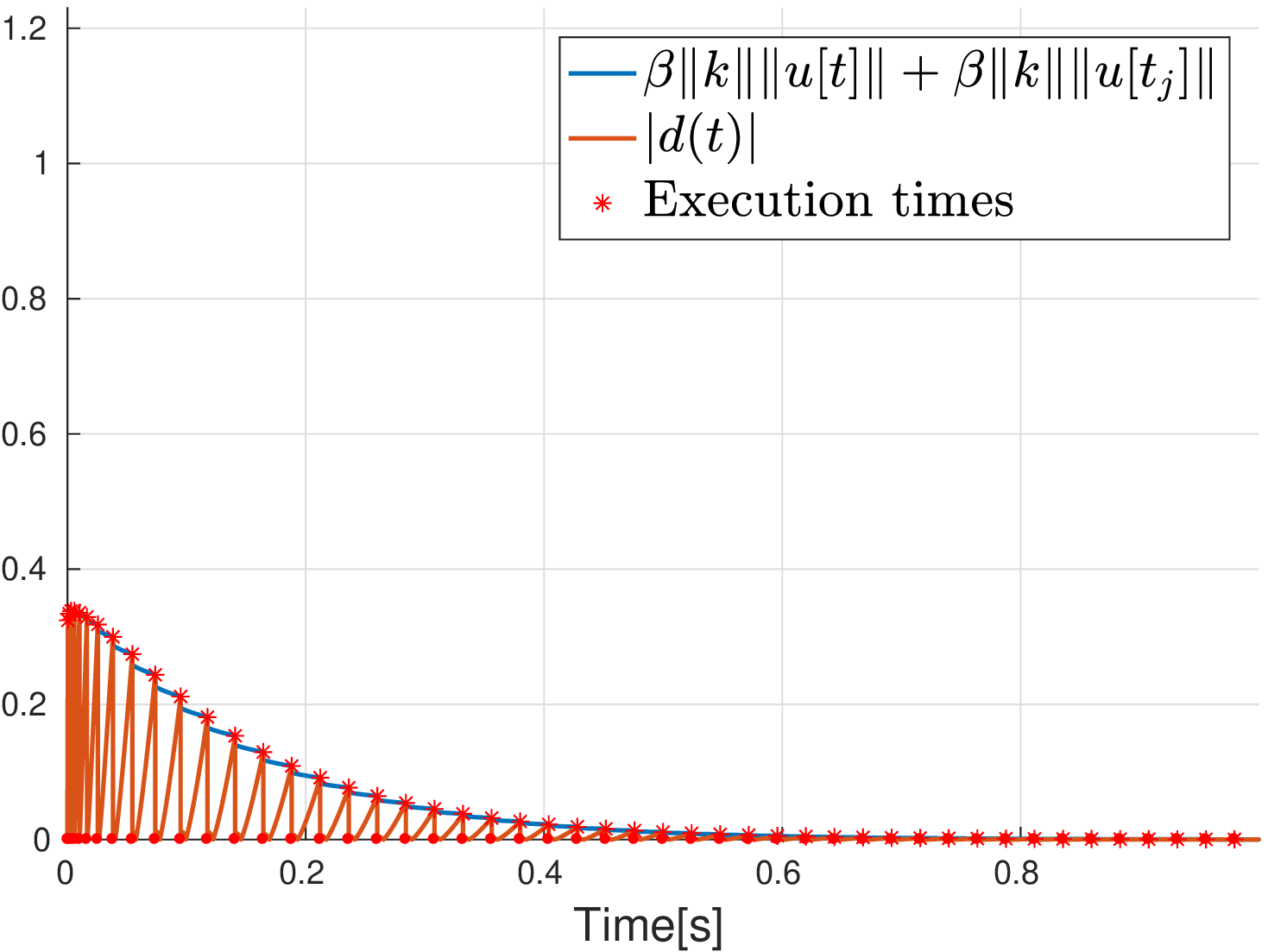} }%scaling font 120% in matlab
\caption{Trajectories involved in the triggering condition \eqref{triggering_conditionISS_with_backstepping_original} (on the left with  $\beta=0.07$ and on the right with  $\beta=0.02$, resulting in slow and fast sampling, respectively). Once the trajectory $ \vert d \vert$ reaches the trajectory $\beta \Vert  k \Vert \Vert u[t] \Vert +\beta \Vert  k \Vert \Vert u[t_{j}] \Vert $, an event is generated, the control value is updated
and $d$ is reset to zero. }
\label{trajectories_triggering_condition}
}
\end{figure*}
\vskip 0.5cm 
\noindent Finally, we run simulations for 100 different initial conditions given by  $u_{0}(x)  = \sum_{n=1}^l n^2\sqrt{2}\sin(n\pi x) + l(x^2 - x^3)$ for    $l=1,..,10$  and $u_{0}(x) = \sum_{n=1}^l n\sqrt{2}\sin(n^2\pi x) + l(x^2 - x^3)$, for  $l=11,...,100$ on a  frame of $1s$.  We have computed the inter-execution times between two triggering times. 
We compared the cases for slow and fast sampling, i.e. when $\beta = 0.07 $ and $\beta=0.02$, respectively. Figure \ref{histograms} shows the density of the inter-execution times plotted in logarithmic scale where  it can be observed that,  the larger $\beta$ the less often  is the sampling and control updating  which in turn implies larger inter-executions times.

It is interesting to notice that when choosing $\beta$ small (resulting in  fast sampling, as aforementioned), there are several inter-execution times  of the order of $10^{-1.7}$ as depicted in blue bars in Figure \ref{histograms} where the density predominates. It might suggest that a possible period (whenever one intends to sample periodically in a sampled-and-hold fashion) might be chosen with a   length of the order $10^{-1.7}$.  This issue is left for further  tests and investigation with possible theoretical connections with periodic schemes as in \cite{KARAFYLLIS2018226}. This issue may give some hints on how  to suitably choose sampling periods in order to reduce conservatism on periodic schemes.
%{\color{red} However, we want to sample the less often? }

% {\color{blue}It is important to emphasizes that in both cases,  from such computations the minimal inter-executions times are ...  and ... that are larger than the time step of discretization}.

%\begin{figure*}[t]
%%\captionsetup[subfigure]{labelformat=empty}
%\centering{
%\subfigure{\includegraphics[width=0.5\columnwidth]{histograma_slow.eps} }%scaling font 120% in matlab
%\subfigure{\includegraphics[width=0.5\columnwidth]{histograma_fast.eps} }
%\caption{Density of the inter-execution times: with  $\beta=0.07$  (left) for slow sampling and with   $\beta=0.02$ (right) for fast sampling.}
%\label{histograms}
%}
%\end{figure*}

\begin{figure*}[t]
%\captionsetup[subfigure]{labelformat=empty}
\centering{
\includegraphics[width=0.9\textwidth]{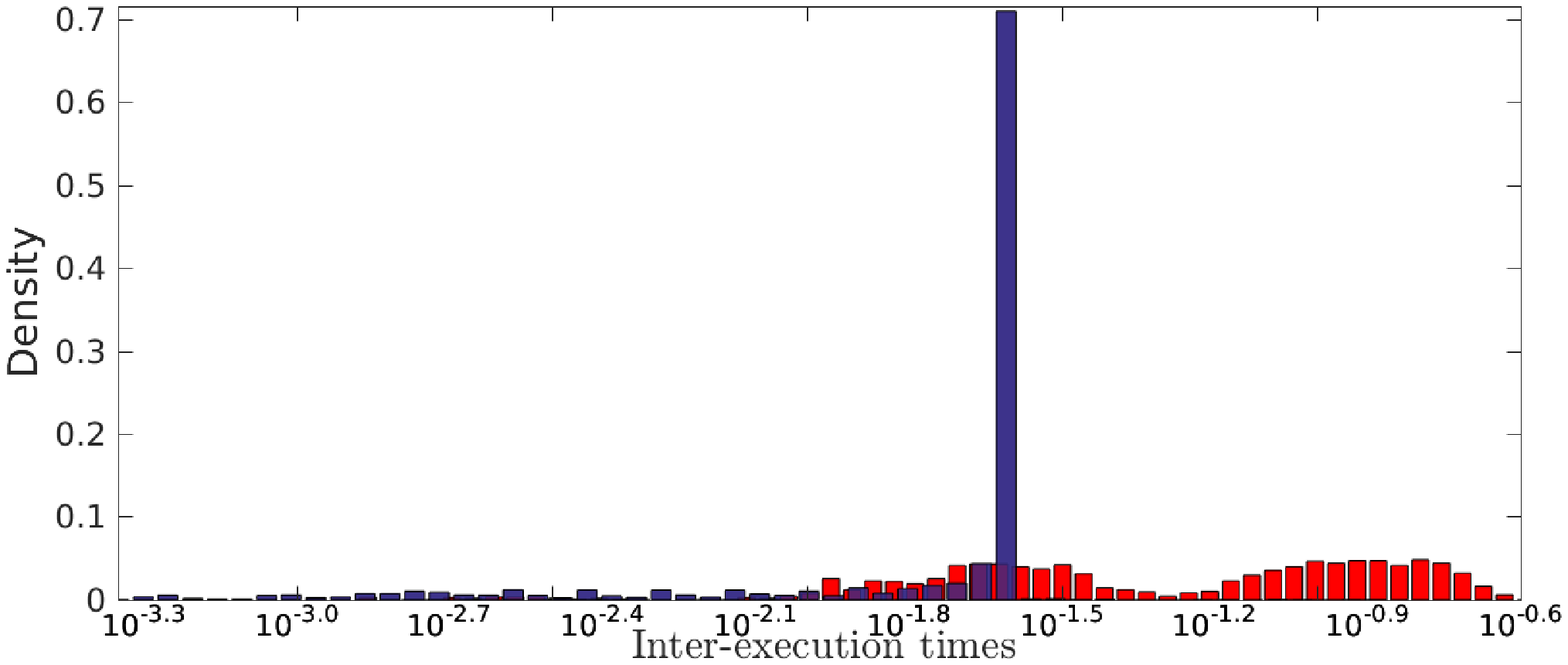} }%scaling font 120% in matlab

\caption{Density of the inter-execution times (logarithmic scale) computed for $100$ different initial conditions given by  $u_{0}(x)  = \sum_{n=1}^l n^2\sqrt{2}\sin(n\pi x) + l(x^2 - x^3)$ for    $l=1,..,10$  and $u_{0}(x) = \sum_{n=1}^l n\sqrt{2}\sin(n^2\pi x) + l(x^2 - x^3)$, for  $l=11,...,100$ on a  frame of $1s$. With  $\beta=0.07$  implying  slow sampling and therefore larger inter-execution times (red bars) and with   $\beta=0.02$ implying  fast sampling and therefore smaller inter-execution times (blue bars). }
\label{histograms}
\end{figure*}

% \subsubsection*{On the computation of the minimal dwell time and Comparison with period obtained in  \cite{KARAFYLLIS2018226} }
%which is instrumental for the computation of the minimal
%
%dwell-time
% Note that the dwell times are very small and comparable with period that can be obtained from \cite{KARAFYLLIS2018226}. This strong conservativ relies on the fact of the analysis was done from  small gzain arguments along wit conservaite upper estimates.
%
%
%N has to be suitable chosen: one the one hand, sufficiently large to
%make a 2 small enough; but not that large because a 1 would be huge.
%The existence of T is certainly guaranteed in order to meet Φ p < 1;
%however, it turns out to be very very small (observation from several
%different numerical simulations with different parameters).
%In practice a sampling period very very small makes the periodic
%implementation very conservative implying that one would require high
%utilization of computational resources

\section{Conclusion}\label{conslusion_and_perspect}

%As motivated throughout the paper, with this event-triggered approach we can reduce significantly the number of execution
%times while updating the control value only when needed and while ensuring the theoretical guarantees. 
%

In this paper, we have proposed an event-triggered boundary control to stabilize (on events) a reaction-diffusion PDE system with Dirichlet boundary condition. A suitable state-dependent event-triggering condition  is considered. It  determines when the control has to be updated. It has been proved the existence of a minimal dwell-time which is independent of the initial condition. Thus, it has been proved  that there is no Zeno behavior and thus the well-posedness and the stability of the closed-loop system 
are guaranteed.

In future work, we may consider observer-based event-triggered control and possibly  sampling   output measurements  on events as well. It may suggest that another event-triggered  strategy shall be considered to be combined with the one for actuation. 
We expect also to  address periodic event-triggered strategies inspired by some recent result from finite-dimensional systems \cite{BORGERS201881}. For that, we may use the obtained dwell-time as a period or to come up with a maybe less conservative period. In either cases, the period would be utilized to monitor periodically the triggering condition  while the actuation is still on events. This would represent even a more realistic approach toward digital realizations while  reducing the consumption of computational resources.

%\section*{Acknowledgements}

\bibliographystyle{plain}
\bibliography{ETC_heat}

\begin{thebibliography}{10}

\bibitem{BastinCoron_book2016}
G.~Bastin and J.-M. Coron.
\newblock {\em Stability and Boundary Stabilization of 1-D Hyperbolic Systems}.
\newblock {Birkh\"auser} Basel, 2016.

\bibitem{BORGERS201881}
D.P. Borgers, R.~Postoyan, A.~Anta, P.~Tabuada, D.~{Nešić}, and W.P.M.H.
  Heemels.
\newblock Periodic event-triggered control of nonlinear systems using
  overapproximation techniques.
\newblock {\em Automatica}, 94:81 -- 87, 2018.

\bibitem{Corless1996}
R.M. Corless, G.H. Gonnet, D.E.G. Hare, D.J. Jeffrey, and D.E Knuth.
\newblock On the {Lambert} {W} function.
\newblock {\em Advances in Computational Mathematics}, 5:329--359, 1996.

\bibitem{Davosampling2018}
MA. Davo, D.~Bresch-Pietri, C~Prieur, and F.~Di~Meglio.
\newblock Stability analysis of a $2\times 2$ linear hyperbolic system with a
  sampled-data controller via backstepping method and looped-functionals.
\newblock {\em IEEE Transactions on Automatic Control}, 2018.

\bibitem{Espitia2016_Aut}
N.~Espitia, A.~Girard, N.~Marchand, and C~Prieur.
\newblock Event-based control of linear hyperbolic systems of conservation
  laws.
\newblock {\em Automatica}, 70:275--287, August 2016.

\bibitem{Espitia2016Nolcos}
N.~Espitia, A.~Girard, N.~Marchand, and C.~Prieur.
\newblock Event-based stabilization of linear systems of conservation laws
  using a dynamic triggering condition.
\newblock In {\em Proc. of the 10th IFAC Symposium on Nonlinear Control Systems
  (NOLCOS)}, volume~49, pages 362--367, Monterey (CA), USA, 2016.

\bibitem{Espitia2018TAC}
N.~Espitia, A.~Girard, N.~Marchand, and C~Prieur.
\newblock Event-based boundary control of a linear 2x2 hyperbolic system via
  backstepping approach.
\newblock {\em IEEE Transactions on Automatic Control}, 63(8):2686--2693, 2018.

\bibitem{Fridman2012826}
E.~Fridman and A.~Blighovsky.
\newblock Robust sampled-data control of a class of semilinear parabolic
  systems.
\newblock {\em Automatica}, 48(5):826--836, 2012.

\bibitem{girard2014dynamic}
A.~Girard.
\newblock Dynamic triggering mechanisms for event-triggered control.
\newblock {\em IEEE Transactions on Automatic Control}, 60(7):1992--1997, July
  2015.

\bibitem{event-triger-Heeme-Johan-Tabu}
W.P.M.H. Heemels, K.H. Johansson, and P.~Tabuada.
\newblock An introduction to event-triggered and self-triggered control.
\newblock In {\em Proceedings of the 51st IEEE Conference on Decision and
  Control}, pages 3270--3285, Maui, Hawaii, 2012.

\bibitem{Hetel2017309}
L.~Hetel, C.~Fiter, H.~Omran, A.~Seuret, E.~Fridman, J.-P. Richard, and SI.
  Niculescu.
\newblock Recent developments on the stability of systems with aperiodic
  sampling: An overview.
\newblock {\em Automatica}, 76:309 -- 335, 2017.

\bibitem{JIANG20162854}
A.~Jiang, B.~Cui, W.~Wu, and B.~Zhuang.
\newblock Event-driven observer-based control for distributed parameter systems
  using mobile sensor and actuator.
\newblock {\em Computers \& Mathematics with Applications}, 72(12):2854 --
  2864, 2016.

\bibitem{JiangSmallGainETC}
Z.-P Jiang, T.~Liu, and P.~Zhang.
\newblock Event-triggered control of nonlinear systems: A small-gain paradigm.
\newblock In {\em 13th IEEE International Conference on Control Automation
  (ICCA)}, pages 242--247, July 2017.

\bibitem{Sampled_dataKarafylis_Kristic}
I.~Karafyllis and M.~Krstic.
\newblock Sampled-data boundary feedback control of {1-D Hyperbolic PDEs} with
  non-local terms.
\newblock {\em Systems \& Control Letters}, 17:68--75, 2017.

\bibitem{KARAFYLLIS2018226}
I.~Karafyllis and M.~Krstic.
\newblock Sampled-data boundary feedback control of {1-D parabolic PDEs}.
\newblock {\em Automatica}, 87:226 -- 237, 2018.

\bibitem{Karafyllis2019_book}
I.~Karafyllis and M.~Krstic.
\newblock {\em Input-to-State Stability for PDEs}.
\newblock Springer-Verlag, London (Series: Communications and Control
  Engineering), 2019.

\bibitem{Karafyllis_Krstic_SIAM2019}
I.~Karafyllis and M.~Krstic.
\newblock Small-gain-based boundary feedback design for global exponential
  stabilization of 1-d semilinear parabolic pdes.
\newblock {\em SIAM Journal on Control and Optimization}, 57(3):2016--2036,
  2019.

\bibitem{Karafyllis_Adaptive-regulation-triggered2019}
I.~Karafyllis, M.~Krstic, and K.~Chrysafi.
\newblock Adaptive boundary control of constant-parameter reaction–diffusion
  pdes using regulation-triggered finite-time identification.
\newblock {\em Automatica}, 103:166--179, 2019.

\bibitem{krstic2008backstepping}
M.~Krstic and A.~Smyshlyaev.
\newblock Backstepping boundary control for first-order hyperbolic {PDEs} and
  application to systems with actuator and sensor delays.
\newblock {\em Systems \& Control Letters}, 57(9):750--758, 2008.

\bibitem{krstic2008boundary}
M.~Krstic and A.~Smyshlyaev.
\newblock {\em Boundary control of PDEs: A course on backstepping designs},
  volume~16.
\newblock Siam, 2008.

\bibitem{lemmon2010event}
M.~Lemmon.
\newblock Event-triggered feedback in control, estimation, and optimization.
\newblock In {\em Networked Control Systems}, pages 293--358. Springer, 2010.

\bibitem{Liu_ZPJiang2015}
T.~Liu and Z.-P. Jiang.
\newblock A small-gain-approach to robust event-triggered control of nonlinear
  systems.
\newblock {\em IEEE Transaction on Automatic Control}, 60(8):2072--2085, 2015.

\bibitem{Logemann2005}
H.~Logemann, R.~Rebarber, and S.~Townley.
\newblock Generalized sampled-data stabilization ofwell-posed linear
  infinite-dimensional systems.
\newblock {\em SIAM Journal on Control and Optimization}, 44:1345--1369, 2005.

\bibitem{Marchand2013-universal-formula}
N.~Marchand, S.~Durand, and J.~F.~G. Castellanos.
\newblock A general formula for event-based stabilization of nonlinear systems.
\newblock {\em IEEE Transactions on Automatic Control}, 58(5):1332--1337, 2013.

\bibitem{Postoyan_Aframework_ETS2014}
R.~Postoyan, P.~Tabuada, D.~{Nešić}, and A.~Anta.
\newblock A framework for the event-triggered stabilization of nonlinear
  systems.
\newblock {\em IEEE Transactions on Automatic Control}, 60(4):982--996, 2015.

\bibitem{Selivanov_FridmanAuto}
A.~Selivanov and E.~Fridman.
\newblock Distributed event-triggered control of transport-reaction systems.
\newblock {\em Automatica}, 68:344--351, 2016.

\bibitem{Seuret2016}
A.~Seuret, S.~Tarbouriech, C.~Prieur, and L.~Zaccarian.
\newblock {LQ}-based event-triggered controller co-design for saturated linear
  systems.
\newblock {\em Automatica}, 74:47--54, 2016.

\bibitem{Smyshlyaev-Krstic2004}
A.~Smyshlyaev and M.~Krstic.
\newblock Closed-form boundary state feedbacks for a class of 1-d partial
  integro-differential equations.
\newblock {\em IEEE Transactions on Automatic Control,}, 49(12):2185--2202, Dec
  2004.

\bibitem{tabuada2007event}
P.~Tabuada.
\newblock Event-triggered real-time scheduling of stabilizing control tasks.
\newblock {\em IEEE Transactions on Automatic Control}, 52(9):1680--1685, 2007.

\bibitem{Tan2009}
Y.~Tan, E.~{Trélat}, Y.~Chitour, and D.~{Nešić}.
\newblock Dynamic practical stabilization of sampled-data linear distributed
  parameter systems.
\newblock In {\em IEEE 48th Conference on Decision and Control (CDC)}, pages
  5508--5513, 2009.

\bibitem{Vazquez2011}
R.~Vazquez, M.~Krstic, and J.-M. Coron.
\newblock Backstepping boundary stabilization and state estimation of a $2
  \times 2$ linear hyperbolic system.
\newblock In {\em the 50th IEEE Conference on Decision and Control and European
  Control Conference (CDC-ECC)}, pages 4937--4942, Orlando, United States,
  2011.

\bibitem{Yao2013}
Z.~Yao and N.H. El-Farra.
\newblock Resource-aware model predictive control of spatially distributed
  processes using event-triggered communication.
\newblock In {\em Proceedings of the 52nd IEEE Conference on Decision and
  Control}, pages 3726--3731, Florence, Italy, 2013.

\bibitem{Yi2010}
S.~Yi, P.W. Nelson, and A.G. Ulsoy.
\newblock {\em Time-Delay Systems: Analysis and Control Using the Lambert W
  Function}.
\newblock World Scientific Publishing, Singapore, 2010.

\end{thebibliography}

\end{document}